\documentclass[submission,copyright,creativecommons,sharealike,noncommercial]{eptcs}
\providecommand{\event}{ITRS 2018}

\usepackage{float}
\usepackage{amsmath}
\usepackage{amssymb}
\usepackage{amsthm}
\usepackage{stmaryrd}
\usepackage{bussproofs}
\usepackage{graphicx}

\newcommand{\online}[1]{url: \url{#1}\;}

\floatstyle{boxed}
\restylefloat{table}
\restylefloat{figure}

\newtheoremstyle{break}{}{}
                       {\normalsize\normalfont\itshape}
                       {}
                       {\normalsize\normalfont\bfseries}
                       {}{\newline}
                       {}
\newtheoremstyle{breakrm}{}{}
                         {\normalsize\normalfont\rmfamily}
                         {}
                         {\normalsize\normalfont\bfseries}
                         {}{\newline}
                         {}
\theoremstyle{break}
\newtheorem{thm}{Theorem}
\newtheorem{prop}{Proposition}
\newtheorem{lem}{Lemma}
\theoremstyle{breakrm}

\newcommand*{\ie}{\emph{i.e.}}
\newcommand*{\logsys}[1]{\textnormal{\textsf{#1}}}
\newcommand*{\BCD}{\logsys{BCD}}
\newcommand*{\Coq}{\texttt{\upshape Coq}}

\newcommand*{\sem}[1]{\llbracket #1\rrbracket}
\newcommand*{\subst}[3]{#1[{^{#2}/_{#3}}]}

\newcommand*{\lb}[2]{\lambda #1.#2}
\newcommand*{\ap}[2]{#1\,#2}
\newcommand*{\pa}[2]{\langle #1,#2\rangle}
\newcommand*{\pl}[1]{\pi_1\,#1}
\newcommand*{\pr}[1]{\pi_2\,#1}
\newcommand*{\imp}{\rightarrow}
\newcommand*{\inter}{\cap}
\newcommand*{\biginter}{\bigcap}
\newcommand*{\omg}{\Omega}
\newcommand*{\reduct}{\rightarrow}
\newcommand*{\bred}{\reduct_\beta}
\newcommand*{\pred}{\reduct_\pi}

\newcommand*{\var}{\textit{var}}
\newcommand*{\abs}{\textit{abs}}
\newcommand*{\app}{\textit{app}}
\newcommand*{\pair}{\textit{pair}}
\newcommand*{\projl}{\textit{proj}\ensuremath{_1}}
\newcommand*{\projr}{\textit{proj}\ensuremath{_2}}

\newcommand*{\ax}{\ensuremath{\textit{ax}}}
\newcommand*{\cut}{\ensuremath{\textit{cut}}}
\newcommand*{\exc}{\ensuremath{\textit{ex}}}
\newcommand*{\refl}{\textit{refl}}
\newcommand*{\trans}{\textit{trans}}
\newcommand*{\axomgr}{\ensuremath{\omg_r}}
\newcommand*{\axinterll}{\ensuremath{\inter_l^1}}
\newcommand*{\axinterlr}{\ensuremath{\inter_l^2}}
\newcommand*{\axinterr}{\ensuremath{\inter_r}}

\newcommand*{\A}{\Gamma}
\newcommand*{\B}{\Delta}
\newcommand*{\C}{\Sigma}
\newcommand*{\AIC}[1]{\AxiomC{\ensuremath{#1}}}
\newcommand*{\ZIC}[1]{\AIC{}\UIC{#1}}
\newcommand*{\UIC}[1]{\UnaryInfC{\ensuremath{#1}}}
\newcommand*{\BIC}[1]{\BinaryInfC{\ensuremath{#1}}}
\newcommand*{\TIC}[1]{\TrinaryInfC{\ensuremath{#1}}}
\newcommand*{\QIC}[1]{\QuaternaryInfC{\ensuremath{#1}}}
\newcommand*{\QuIC}[1]{\QuinaryInfC{\ensuremath{#1}}}
\newcommand*{\DP}{\DisplayProof}
\newcommand*{\RL}[1]{\RightLabel{\ensuremath{#1}}}
\newcommand{\newrule}[3]{\newcommand{#1}[1]{\RL{#2}\csname #3IC\endcsname{##1}}}
\newrule{\TRANS}{\trans}{B}
\newrule{\VAR}{\var}{Z}
\newrule{\ABS}{\abs}{U}
\newrule{\APP}{\app}{B}
\newrule{\PAIR}{\pair}{B}
\newrule{\PROJL}{\projl}{U}
\newrule{\PROJR}{\projr}{U}
\newrule{\INTER}{\inter}{B}
\newrule{\OMG}{\omg}{Z}
\newrule{\LEQ}{\leq}{B}
\newrule{\REFL}{\refl}{Z}
\newrule{\IMP}{\imp}{B}

\newcommand*{\IScons}{\logsys{ISC}}
\newcommand*{\constr}{\kappa}
\newcommand*{\cprec}{\preccurlyeq}
\newcommand*{\convar}[1]{\alpha_{#1}}
\newcommand*{\covar}[1]{\beta_{#1}}
\newcommand*{\prewidth}{\varepsilon}
\newcommand*{\width}[1]{\prewidth({#1})}
\newcommand*{\wk}{\textit{wk}}
\newcommand*{\co}{\textit{co}}
\newcommand*{\wkg}{\ensuremath{\wk_{gen}}}
\newcommand*{\constrr}{\textit{constr}}
\newcommand*{\interr}{\ensuremath{{\inter}R}}
\newcommand*{\interl}{\ensuremath{{\inter}L}}
\newcommand*{\interle}{\ensuremath{{\inter}L_e}}
\newrule{\ISAX}{\ax}{Z}
\newrule{\ISOMG}{\omg}{Z}
\newrule{\ISINTERR}{{\inter}R}{B}
\newrule{\ISINTERLL}{\interl_1}{U}
\newrule{\ISINTERLR}{\interl_2}{U}
\newrule{\ISIMPR}{{\imp}R}{U}
\newrule{\ISIMPL}{{\imp}L}{B}
\newrule{\ISCUTL}{\cut_1}{B}
\newrule{\ISCUTR}{\cut_2}{B}
\newrule{\ISCINTERR}{\interr}{B}
\newrule{\ISCINTERL}{\interl}{U}
\newrule{\ISCWK}{\wk}{U}
\newrule{\ISCC}{\constrr}{T}
\newrule{\ISCWKg}{\wkg}{U}
\newrule{\ISCEXC}{\exc}{U}
\newrule{\ISCAX}{\ax}{Z}
\newrule{\ISCINTERLe}{\interle}{U}
\newrule{\ISCCO}{\co}{U}
\newrule{\ISCCUT}{\cut}{B}
\newrule{\ISCAXX}{\ax_\textit{at}}{Z}
\newrule{\ISCOMG}{\omg}{Z}
\newrule{\ISCIMP}{\imp}{Q}
\newrule{\ISCIMPZ}{\imp_0}{U}
\newrule{\ISCTIMES}{\times}{B}

\title{Intersection Subtyping with Constructors}
\author{Olivier Laurent\thanks{This work was supported by the LABEX MILYON (ANR-10-LABX-0070) of Universit\'e de Lyon, within the program ``Investissements d'Avenir'' (ANR-11-IDEX-0007), and by the project Elica (ANR-14-CE25-0005), both operated by the French National Research Agency (ANR). This work was also supported by GDRI Linear Logic.}
\institute{Univ Lyon, EnsL, UCBL, CNRS,  LIP, F-69342, LYON Cedex 07, France}
\email{olivier.laurent@ens-lyon.fr}
}
\def\titlerunning{Intersection Subtyping with Constructors}
\def\authorrunning{O. Laurent}

\begin{document}
\maketitle

\begin{abstract}
We study the question of extending the \BCD{} intersection type system with additional type constructors. On the typing side, we focus on adding the usual rules for product types. On the subtyping side, we consider a generic way of defining a subtyping relation on families of types which include intersection types. We find back the \BCD{} subtyping relation by considering the particular case where the type constructors are intersection, omega and arrow. We obtain an extension of \BCD{} subtyping to product types as another instance. We show how the preservation of typing by both reduction and expansion is satisfied in all the considered cases. Our approach takes benefits from a ``subformula property'' of the proposed presentation of the subtyping relation.
\end{abstract}

% Keywords: intersection types ; subtyping ; product types ; cut elimination ; subject reduction

\section{Introduction}

Intersection type systems are tools for building and analysing models of the $\lambda$-calculus~\cite{interbcd,interbakel,paramlambda,intermodels}. They also provide ways of characterising reduction properties of $\lambda$-terms such as normalization.
The main difference to other type systems is the fact that not only \emph{subject reduction} holds (if $t$ reduces to $u$ and $\A\vdash t:A$ then $\A\vdash u:A$) but also \emph{subject expansion} holds (if $t$ reduces to $u$ and $\A\vdash u:A$ then $\A\vdash t:A$). As a consequence it is possible to define a denotational model by associating to each (closed) term the set of its types $\sem{t}=\{A\mid{}\vdash t:A\}$.

The most famous intersection type system is probably the \BCD{} system~\cite{interbcd}, and this is the one we are focusing on.
While \BCD{} insists on the interaction between arrow types and intersection types, we want to consider more general sets of type constructors, following~\cite{intermixin}. The \BCD{} type system can be decomposed into two parts: typing rules and subtyping rules. They are related through the subsumption rule. Our main contribution is a derivation system for the subtyping relation which allows us to deal with generic type constructors while satisfying a ``subformula property''.
Alternative presentations of \BCD{} subtyping are studied in~\cite{interlf}.
In contrast with ~\cite{intermixin}, we allow contravariant type constructors so that even the arrow constructor can be defined as an instance of our generic pattern, and only intersection has a specific status.

\bigskip

In Section~\ref{sectyp}, we recall standard syntactic proofs~\cite{intermodels,tutinter} of preservation of typing by $\beta$-reduction and $\beta$-expansion for the \BCD{} system. Our presentation stresses the fact that, starting from intersection (and $\omg$) only, type constructors can be added in a modular way. In Section~\ref{sectypimp}, we consider the arrow types, thus obtaining the usual \BCD{} rules. We extend the results to product types in Section~\ref{sectypprod}. The main part of the paper is then Section~\ref{secsub} where we propose a sequent-style derivation system for defining \BCD-like subtyping relations for extensions of intersection types to generic sets of constructors. Starting from a transitivity/cut admissibility property, we prove that instances of our system are equivalent with variants of the \BCD{} subtyping relation.

Key results on subtyping (Propositions~\ref{propisconsadmiss} and~\ref{propmin}, and Theorem~\ref{thmisconsbcd}) are formalized in \Coq:
\begin{center}
\url{https://perso.ens-lyon.fr/olivier.laurent/bcdc/}
\end{center}

\section{Intersection Typing}\label{sectyp}

We present the system we are looking at, which is mainly \BCD~\cite{interbcd} extended with product types. Type constructors are introduced in an incremental and modular way.

\subsection{Intersection Types}

Let us first consider an at most countable set $\mathcal{X}$ of base types denoted $X$, $Y$, etc, and consider types built using at least the following constructors:
\begin{equation*}
  A,B ::= X \mid A\inter B \mid \omg \mid \dotsc
\end{equation*}
Similarly we do not define the exact set of terms (denoted $t$, $u$, etc), but first we only assume they contain a denumerable set of term variables $\mathcal{V}$ (whose elements are denoted $x$, $y$, etc).
A first set of typing rules is given on Table~\ref{tabtyp}, with judgements $\A\vdash t:A$ built from a list $\A$ of typing declarations for variables (of the shape $x:B$), a term $t$ and a type $A$.
Note these rules rely on a \emph{subtyping} relation $\leq$ on types, which is here just a parameter.

\begin{table}
\centering
\begin{gather*}
\VAR{\A,x:A,\B\vdash x:A}
\DP
\qquad
\AIC{\A\vdash t:A}
\AIC{A\leq B}
\LEQ{\A\vdash t:B}
\DP
\qquad
\AIC{\A\vdash t:A}
\AIC{\A\vdash t:B}
\INTER{\A\vdash t:A\inter B}
\DP
\qquad
\OMG{\A\vdash t:\omg}
\DP
\end{gather*}
\caption{Typing Rules with Subtyping and Intersection.}\label{tabtyp}
\end{table}

\begin{lem}[Weakening]\label{lemwk}
If $\A\vdash t:A$ and $\B\leq\A$ (meaning that, for each $x:B$ in $\A$, one can find $x:B'$ in $\B$ with $B'\leq B$) then $\B\vdash t:A$.
\end{lem}

\begin{lem}[Strengthening]\label{lemstr}
If $\A,x:B,\B\vdash t:A$ and $x\notin t$ then $\A,\B\vdash t:A$.
\end{lem}

Because it makes hypotheses on the term in conclusion, the rule ($\var$) is called a \emph{term rule} (the introduced term must be a variable). In the opposite, ($\leq$), ($\inter$) and ($\omg$) rules are called \emph{non-term} as they apply on any term without any constraint on its main constructor.
As a term rule, ($\var$) admits a so-called \emph{generation lemma} analysing how variables can by typed. For this, we make some hypotheses on the subtyping relation (see Table~\ref{tabsub}).

Note in passing, that the axioms of Table~\ref{tabsub} are equivalent to say that $\leq$ is a preorder relation with $\inter$ as greatest lower bound and $\omg$ as top element. In particular, up to the equivalence relation induced by $\leq$, $\inter$ is a commutative associative idempotent operation with $\omg$ as unit. As a consequence the notation $\biginter_{i\in I}A_i$ makes sense (up to the equivalence relation induced by $\leq$) for any (possibly empty) finite set $I$.

\begin{table}
\centering
\begin{align}
  A &\leq A \tag{\refl}\label{axsubrefl}\\
  A\leq B \;\wedge\; B\leq C \;&\Rightarrow\; A\leq C \tag{\trans}\label{axsubtrans}\\
  A\inter B &\leq A \tag{\axinterll}\label{axsubinterll} \\
  A\inter B &\leq B \tag{\axinterlr}\label{axsubinterlr} \\
  C\leq A \;\wedge\; C\leq B \;&\Rightarrow\; C\leq A\inter B \tag{\axinterr}\label{axsubinterr}\\
  A &\leq \omg \tag{\axomgr}\label{axsubomgr}
\end{align}
\caption{Kernel Properties of Subtyping.}\label{tabsub}
\end{table}

\begin{lem}[Generation for Variables]\label{lemgenvar}
Assuming that ($\var$) is the only term rule introducing a variable, the only non-term rules are ($\leq$), ($\inter$) and ($\omg$), and that the axioms of Table~\ref{tabsub} are satisfied, we have:
if $\A\vdash x:A$ with $x:B\in\A$ then $B\leq A$.
\end{lem}

\begin{lem}[Substitution]\label{lemsubst}
  If $\A,x:A,\B\vdash t:B$ and $\A,\B\vdash u:A$ then $\A,\B\vdash\subst{t}{u}{x}:B$.
\end{lem}

\subsection{Arrow Types}\label{sectypimp}

We now assume types contain an arrow constructor and terms are extended correspondingly:
\begin{equation*}
  A,B ::= X \mid A\inter B \mid \omg \mid A\imp B \mid \dotsc
\hspace{2cm}
  t,u ::= x \mid \lb{x}{t} \mid \ap{t}{u} \mid \dotsc
\end{equation*}
The associated typing rules are given on Table~\ref{tabtypimp}. Note the two new rules are term rules corresponding respectively to $\lb{x}{t}$ and $\ap{t}{u}$ (no new non-term rule).

\begin{table}
\centering
\begin{equation*}
\AIC{\A,x:A\vdash t:B}
\ABS{\A\vdash \lb{x}{t}:A\imp B}
\DP
\qquad\qquad
\AIC{\A\vdash t:A\imp B}
\AIC{\A\vdash u:A}
\APP{\A\vdash \ap{t}{u}:B}
\DP
\end{equation*}
\caption{Typing Rules for Arrow.}\label{tabtypimp}
\end{table}

By adding new cases corresponding to the added rules in the proofs, one can check that Lemmas~\ref{lemwk} and~\ref{lemstr}
still hold. Moreover the hypotheses of Lemma~\ref{lemgenvar} are still satisfied, and finally Lemma~\ref{lemsubst} (which only relies on the previous lemmas) is still true as well. Since terms may now contain binders (such as $\lb{x}{\_}$), substitution has to be considered as being capture-avoiding substitution.

\begin{lem}[Generation for Application]\label{lemgenapp}
Assuming that ($\app$) is the only term rule introducing an application, the only non-term rules are ($\leq$), ($\inter$) and ($\omg$), and that the axioms of Table~\ref{tabsub} are satisfied, we have:
if $\A\vdash\ap{t}{u}:B$ then there exist two families of types $(A_i)_{i\in I}$ and $(B_i)_{i\in I}$ with $\biginter_{i\in I}B_i\leq B$ and, for each $i\in I$, $\A\vdash t:A_i\imp B_i$ and $\A\vdash u:A_i$.
\end{lem}

\begin{lem}[Generation for Abstraction]\label{lemgenabs}
Assuming that ($\abs$) is the only term rule introducing an abstraction, the only non-term rules are ($\leq$), ($\inter$) and ($\omg$), and that the axioms of Table~\ref{tabsub} are satisfied, we have:
if $\A\vdash\lb{x}{t}:A$ then there exist two families of types $(B_i)_{i\in I}$ and $(C_i)_{i\in I}$ with $\biginter_{i\in I}(B_i\imp C_i)\leq A$ and, for each $i\in I$, $\A,x:B_i\vdash t:C_i$.
\end{lem}

We now have the requested material to prove subject reduction and subject expansion. However a specific property on subtyping is still missing:
\begin{equation*}
  \biginter_{i\in I}(A_i\imp B_i)\leq A\imp B\;\Rightarrow\; \exists J\subseteq I,\;\left(\biginter_{i\in J}B_i\leq B\;\wedge\;\forall i\in J,\;A\leq A_i\right) \tag{\ensuremath{\imp\leq\imp}}\label{aximpinv}
\end{equation*}
The study of this property will be at the heart of Section~\ref{secsub}.

\begin{prop}[Subject Reduction]\label{propbetared}
  Assuming (\ref{aximpinv}), if $t_1\bred t_2$ and $\A\vdash t_1:A$ then $\A\vdash t_2:A$.
\end{prop}

\begin{proof}
The key case is $\ap{(\lb{x}{t})}{u}\bred\subst{t}{u}{x}$.
If $\A\vdash\ap{(\lb{x}{t})}{u}:A$, by Lemma~\ref{lemgenapp}, we have two families $(B_i)_{i\in I}$ and $(C_i)_{i\in I}$ with $\biginter_{i\in I}C_i\leq A$ and, for each $i\in I$, $\A\vdash\lb{x}{t}:B_i\imp C_i$ and $\A\vdash u:B_i$. For each $i\in I$, by Lemma~\ref{lemgenabs}, we have two families $(B'_j)_{j\in J_i}$ and $(C'_j)_{j\in J_i}$ with $\biginter_{j\in J_i}(B'_j\imp C'_j)\leq B_i\imp C_i$ and, for each $j\in J_i$, $\A,x:B'_j\vdash t:C'_j$. By (\ref{aximpinv}), there exists $K_i\subseteq J_i$ such that $B_i\leq B'_j$ ($j\in K_i$) and $\biginter_{j\in K_i}C'_j\leq C_i$.
We conclude by using Lemma~\ref{lemsubst} with $\A\vdash u:B'_j$:
\begin{prooftree}
  \AIC{\dotsb}
  \AIC{\dotsb}
  \AIC{\A\vdash\subst{t}{u}{x}:C'_j}
  \AIC{\dotsb}
  \RL{\inter}
  \TIC{\A\vdash\subst{t}{u}{x}:\biginter_{j\in K_i}C'_j}
  \AIC{\biginter_{j\in K_i}C'_j\leq C_i}
  \LEQ{\A\vdash\subst{t}{u}{x}:C_i}
  \AIC{\dotsb}
  \RL{\inter}
  \TIC{\A\vdash\subst{t}{u}{x}:\biginter_{i\in I}C_i}
  \AIC{\biginter_{i\in I}C_i\leq A}
  \LEQ{\A\vdash\subst{t}{u}{x}:A}
\end{prooftree}
\end{proof}

\begin{prop}[Subject Expansion]\label{propbetaexp}
  If $t_1\bred t_2$ and $\A\vdash t_2:A$ then $\A\vdash t_1:A$.
\end{prop}

\begin{proof}
The key case is $\ap{(\lb{x}{t})}{u}\bred\subst{t}{u}{x}$.
  We first prove that $\A\vdash\subst{t}{u}{x}:B$ implies that we can find a type $A$ such that $\A,x:A\vdash t:B$ and $\A\vdash u:A$,
by induction on the derivation of $\A\vdash\subst{t}{u}{x}:B$.
And then:
  \begin{prooftree}
    \AIC{\A,x:A\vdash t:B}
    \ABS{\A\vdash\lb{x}{t}:A\imp B}
    \AIC{\A\vdash u:A}
    \APP{\A\vdash\ap{(\lb{x}{t})}{u}:B}
  \end{prooftree}
\end{proof}

To sum up, we have shown that, given the typing rules of Tables~\ref{tabtyp} and~\ref{tabtypimp}, the subject reduction and subject expansion properties hold for $\beta$-reduction as soon as the chosen subtyping satisfies the axioms of Table~\ref{tabsub} as well as property~(\ref{aximpinv}).
The historical example from the literature is the \BCD{} system~\cite{interbcd} corresponding to the subtyping relation of Table~\ref{tabsubbcd}. We will come back to the fact that (\ref{aximpinv}) holds for this \BCD{} relation (Lemma~\ref{leminvimp}).

\begin{table}
\centering
\begin{gather*}
\ZIC{A\leq A}
\DP
\qquad\qquad
\AIC{A\leq B}
\AIC{B\leq C}
\BIC{A\leq C}
\DP
\qquad\qquad
\ZIC{A\leq\omg}
\DP
\\[2ex]
\ZIC{A\inter B\leq A}
\DP
\qquad\qquad
\ZIC{A\inter B\leq B}
\DP
\qquad\qquad
\ZIC{A\leq A\inter A}
\DP
\qquad\qquad
\AIC{A\leq C}
\AIC{B\leq D}
\BIC{A\inter B\leq C\inter D}
\DP
\\[2ex]
\AIC{C\leq A}
\AIC{B\leq D}
\BIC{A\imp B\leq C\imp D}
\DP
\qquad\qquad
\ZIC{(A\imp B)\inter(A\imp C)\leq A\imp(B\inter C)}
\DP
\qquad\qquad
\ZIC{\omg\leq\omg\imp\omg}
\DP
\end{gather*}
\caption{\BCD{} Subtyping Rules.}\label{tabsubbcd}
\end{table}

\subsection{Product Types}\label{sectypprod}

We now assume types contain a product constructor and terms are extended correspondingly:
\begin{equation*}
  A,B ::= X \mid A\inter B \mid \omg \mid A\imp B \mid A\times B \mid \dotsc
\hspace{1.5cm}
  t,u ::= x \mid \lb{x}{t} \mid \ap{t}{u} \mid \pa{t}{u} \mid \pl{t} \mid \pr{t} \mid \dotsc
\end{equation*}
The associated typing rules are given on Table~\ref{tabtypprod}. Note the new rules are all term rules (no non-term rule added).
Lemmas~\ref{lemwk}, \ref{lemstr},
\ref{lemgenvar} and~\ref{lemsubst} still hold. It is also easy to check that the new rules do not break Propositions~\ref{propbetared} and~\ref{propbetaexp}.

\begin{table}
\centering
\begin{equation*}
\AIC{\A\vdash t:A}
\AIC{\A\vdash u:B}
\PAIR{\A\vdash\pa{t}{u}:A\times B}
\DP
\qquad\qquad
\AIC{\A\vdash t:A\times B}
\PROJL{\A\vdash\pl{t}:A}
\DP
\qquad\qquad
\AIC{\A\vdash t:A\times B}
\PROJR{\A\vdash\pr{t}:B}
\DP
\end{equation*}
\caption{Typing Rules for Product.}\label{tabtypprod}
\end{table}

\begin{lem}[Generation for Pairing]\label{lemgenpair}
Assuming that ($\pair$) is the only term rule introducing a pair, the only non-term rules are ($\leq$), ($\inter$) and ($\omg$), and that the axioms of Table~\ref{tabsub} are satisfied, we have:
if $\A\vdash\pa{t}{u}:A$ then there exist two families of types $(B_i)_{i\in I}$ and $(C_i)_{i\in I}$ with $\biginter_{i\in I}(B_i\times C_i)\leq A$ and, for each $i\in I$, $\A\vdash t:B_i$ and $\A\vdash u:C_i$.
\end{lem}

\begin{proof}
  By induction on the typing derivation of $\A\vdash\pa{t}{u}:A$, by looking at each possible last rule which, by assumption, can only be ($\pair$), ($\leq$), ($\inter$) or ($\omg$):
  \begin{itemize}
  \item ($\pair$) rule: $I$ is a singleton and the result is immediate.
  \item ($\leq$) rule: we have $\A\vdash\pa{t}{u}:A'$ with $A'\leq A$, we apply the induction hypothesis to $\A\vdash\pa{t}{u}:A'$, and we conclude by ($\trans$).
  \item ($\inter$) rule: we have $A=A'\inter A''$, by induction hypotheses we obtain families of types indexed by $I'$ and $I''$ and we consider $I=I'\uplus I''$. We conclude by using $E_1\leq E_2\;\wedge\; F_1\leq F_2\;\Rightarrow\; E_1\inter F_1\leq E_2\inter F_2$.
  \item ($\omg$) rule: we simply choose $I=\emptyset$.
\qedhere
  \end{itemize}
\end{proof}

\begin{lem}[Generation for Left Projection]\label{lemgenprojl}
Assuming that ($\projl$) is the only term rule introducing a left projection, the only non-term rules are ($\leq$), ($\inter$) and ($\omg$), and that the axioms of Table~\ref{tabsub} are satisfied, we have:
if $\A\vdash\pl{t}:A$ then there exist two families of types $(B_i)_{i\in I}$ and $(C_i)_{i\in I}$ with $\biginter_{i\in I}B_i\leq A$ and, for each $i\in I$, $\A\vdash t:B_i\times C_i$.
\end{lem}

\begin{proof}
  By induction on the typing derivation of $\A\vdash\pl{t}:A$, by looking at each possible last rule which, by assumption, can only be ($\projl$), ($\leq$), ($\inter$) or ($\omg$):
  \begin{itemize}
  \item ($\projl$) rule: $I$ is a singleton and the result is immediate.
  \item ($\leq$) rule: we have $\A\vdash\pl{t}:A'$ with $A'\leq A$, we apply the induction hypothesis to $\A\vdash\pl{t}:A'$, and we conclude by ($\trans$).
  \item ($\inter$) rule: we have $A=A'\inter A''$, by induction hypotheses we obtain families of types indexed by $I'$ and $I''$ and we consider $I=I'\uplus I''$. We conclude by using $E_1\leq E_2\;\wedge\; F_1\leq F_2\;\Rightarrow\; E_1\inter F_1\leq E_2\inter F_2$.
  \item ($\omg$) rule: we simply choose $I=\emptyset$.
\qedhere
  \end{itemize}
\end{proof}

\begin{lem}[Generation for Right Projection]
Assuming that ($\projr$) is the only term rule introducing a right projection, the only non-term rules are ($\leq$), ($\inter$) and ($\omg$), and that the axioms of Table~\ref{tabsub} are satisfied, we have:
if $\A\vdash\pr{t}:A$ then there exist two families of types $(B_i)_{i\in I}$ and $(C_i)_{i\in I}$ with $\biginter_{i\in I}C_i\leq A$ and, for each $i\in I$, $\A\vdash t:B_i\times C_i$.
\end{lem}

\begin{proof}
  Similar to the proof of Lemma~\ref{lemgenprojl}.
\end{proof}

We consider the reduction $\pred$ to be the congruence generated by:
\begin{equation*}
  \pl{\pa{t}{u}}\pred t \hspace{3cm} \pr{\pa{t}{u}}\pred u
\end{equation*}
Similarly to the arrow case, we ask for an additional property of the subtyping relation in order to deduce subject reduction:
\begin{equation*}
  \biginter_{i\in I}(A_i\times B_i)\leq A\times B\;\Rightarrow\;\biginter_{i\in I}A_i\leq A\;\wedge\;\biginter_{i\in I}B_i\leq B  \tag{\ensuremath{\times\leq\times}}\label{axprodinv}
\end{equation*}

\begin{prop}[Subject Reduction for Products]
  Assuming (\ref{axprodinv}), if $t_1\pred t_2$ and $\A\vdash t_1:A$ then $\A\vdash t_2:A$.
\end{prop}

\begin{proof}
The key case is $\pl{\pa{t}{u}}\pred t$.
If $\A\vdash\pl{\pa{t}{u}}:A$, by Lemma~\ref{lemgenprojl}, we have two families $(B_i)_{i\in I}$ and $(C_i)_{i\in I}$ with $\biginter_{i\in I}B_i\leq A$ and, for each $i\in I$, $\A\vdash\pa{t}{u}:B_i\times C_i$. For each $i\in I$, by Lemma~\ref{lemgenpair}, we have two families $(B'_j)_{j\in J_i}$ and $(C'_j)_{j\in J_i}$ with $\biginter_{j\in J_i}(B'_j\times C'_j)\leq B_i\times C_i$ and, for each $j\in J_i$, $\A\vdash t:B'_j$ and $\A\vdash u:C'_j$. By (\ref{axprodinv}), $\biginter_{j\in J_i}B'_j\leq B_i$.
We conclude by:
\begin{prooftree}
  \AIC{\dotsb}
  \AIC{\dotsb}
  \AIC{\A\vdash t:B'_j}
  \AIC{\dotsb}
  \RL{\inter}
  \TIC{\A\vdash t:\biginter_{j\in J_i}B'_j}
  \AIC{\biginter_{j\in J_i}B'_j\leq B_i}
  \LEQ{\A\vdash t:B_i}
  \AIC{\dotsb}
  \RL{\inter}
  \TIC{\A\vdash t:\biginter_{i\in I}B_i}
  \AIC{\biginter_{i\in I}B_i\leq A}
  \LEQ{\A\vdash t:A}
\end{prooftree}
\end{proof}

\begin{prop}[Subject Expansion for Products]
  If $t_1\pred t_2$ and $\A\vdash t_2:A$ then $\A\vdash t_1:A$.
\end{prop}

\begin{proof}
The key case is $\pl{\pa{t}{u}}\pred t$.
We have:
  \begin{prooftree}
    \AIC{\A\vdash t:A}
    \OMG{\A\vdash u:\omg}
    \PAIR{\A\vdash\pa{t}{u}:A\times\omg}
    \PROJL{\A\vdash\pl{\pa{t}{u}}:A}
  \end{prooftree}
\end{proof}

Following~\cite{intermixin} in extending \BCD{} subtyping in the context of additional type constructors, we can consider the rules of Table~\ref{tabsubbcdprod} for subtyping with products. This system satisfies property~(\ref{axprodinv}) (Lemma~\ref{leminvprod}).

\begin{table}
\centering
\begin{equation*}
\AIC{A\leq C}
\AIC{B\leq D}
\BIC{A\times B\leq C\times D}
\DP
\qquad\qquad
\ZIC{(A\times B)\inter(C\times D)\leq (A\inter C)\times(B\inter D)}
\DP
\end{equation*}
\caption{\BCD-Style Subtyping Rules for Products.}\label{tabsubbcdprod}
\end{table}

While the present section focused on the product extension of \BCD, our purpose is to use it as a concrete application of a more general pattern of subtyping between types which include intersection as well as other type constructors. What should be remembered from what we have done so far, is that we can get subject reduction and subject expansion as soon as the subtyping relation satisfies Table~\ref{tabsub} as well as (\ref{aximpinv}) and (\ref{axprodinv}). The next section provides a general approach to these results.

\section{Intersection Subtyping}
\label{secsub}

Inspired by~\cite{intermixin}, we directly consider types built with an arbitrary set of constructors. The case of $\times$ for example will be obtained as a particular instance. We go in fact one step further than~\cite{intermixin} by allowing enough generality in the treatment of constructors so that $\imp$ appears as a constructor among others and not as a specific one as given in~\cite{intermixin}.

\subsection{Generic Subtyping with Constructors}

We consider a given set $\mathcal{K}$ of \emph{type constructors} (denoted $\constr$, $\constr_1$, $\constr_2$, etc) which come with a \emph{contravariant arity} $\convar{\constr}$ and a \emph{covariant arity} $\covar{\constr}$. We assume that arities are respected when constructing types, so that if $\convar{\constr}=2$  and $\covar{\constr}=1$, then $\constr(A,B;C)$ is a type when $A$, $B$ and $C$ are three types. Moreover, for each constructor $\constr$, a $\{0,1\}$-value $\width{\constr}$ defines its behaviour with respect to top types (see below).

Types are thus generated through:
  \begin{equation*}
        A,B ::= A\inter B \mid \constr(\vec A;\vec B)
  \end{equation*}
Base types are provided by constructors with zero arities.

We introduce a sequent-calculus-style derivation system \IScons{} to define the subtyping relation on these types. We will show that applying proof-theoretical methods, such as cut elimination, allows us to deduce easily some properties of subtyping such as Lemma~\ref{leminv}.

Sequents are of the shape $\A\vdash A$ where $\A$ is a (possibly empty) \emph{list} of types.
The intended meaning is:
\begin{equation*}
A_1,\dotsc,A_k\vdash B \text{\quad ``means''\quad}
 A_1\inter\dotsb\inter A_k\leq B
\qquad\qquad\text{(thus if $k=0$, $B$ is a top type)}.
\end{equation*}
The derivation rules are given in Table~\ref{tabiscons} and satisfy the subformula property.

\begin{table}
\centering
\begin{gather*}
\AIC{\A,\B\vdash C}
\ISCWK{\A,\constr(\vec A;\vec B),\B\vdash C}
\DP
\qquad\qquad
\AIC{\A\vdash A}
\AIC{\A\vdash B}
\ISCINTERR{\A\vdash A\inter B}
\DP
\qquad\qquad
\AIC{\A,A,B,\B\vdash C}
\ISCINTERL{\A,A\inter B,\B\vdash C}
\DP
\\[2ex]
\AIC{\begin{array}[b]{c}A_1\vdash A_1^1\quad\dotsb\quad A_1\vdash A_1^k \\ \vdots\\ A_{\convar{\constr}}\vdash A_{\convar{\constr}}^1\quad\dotsb\quad A_{\convar{\constr}}\vdash A_{\convar{\constr}}^k\end{array}}
\AIC{\begin{array}[b]{c}B_1^1,\dotsc,B_1^k\vdash B_1\\ \vdots \\ B_{\covar{\constr}}^1,\dotsc,B_{\covar{\constr}}^k\vdash B_{\covar{\constr}}\end{array}}
\AIC{\width{\constr}\leq k}
\ISCC{\constr(A_1^1,\dotsc,A_{\convar{\constr}}^1;B_1^1,\dotsc,B_{\covar{\constr}}^1),\dotsc,\constr(A_1^k,\dotsc,A_{\convar{\constr}}^k;B_1^k,\dotsc,B_{\covar{\constr}}^k)\vdash \constr(A_1,\dotsc,A_{\convar{\constr}};B_1,\dotsc,B_{\covar{\constr}})}
\DP
\end{gather*}
\caption{\IScons{} Deduction System.}\label{tabiscons}
\end{table}

\begin{prop}[Admissible Rules]\label{propisconsadmiss}
The following rules are admissible in \IScons:
\begin{gather*}
\begin{array}{c}
  \AIC{\A\vdash C}
  \ISCEXC{\A'\vdash C}
  \DP\\[2ex]
  \text{$\A'$ permutation of $\A$}
\end{array}
\qquad\qquad
  \AIC{\A,\B\vdash C}
  \ISCWKg{\A,A,\B\vdash C}
  \DP
\qquad\qquad
  \ISCAX{A\vdash A}
  \DP
\\[4ex]
  \AIC{\A,A\inter B,\B\vdash C}
  \ISCINTERLe{\A,A,B,\B\vdash C}
  \DP
\qquad\qquad
  \AIC{\A,A,A,\B\vdash C}
  \ISCCO{\A,A,\B\vdash C}
  \DP
\qquad\qquad
  \AIC{\A\vdash A}
  \AIC{\B,A,\C\vdash C}
  \ISCCUT{\B,\A,\C\vdash C}
  \DP
\end{gather*}  
\end{prop}

\begin{proof}
($\exc$) is obtained by induction on the proof of the premise.
($\wkg$) is obtained by induction on $A$.
($\ax$) is obtained by induction on $A$ using ($\wkg$).
($\interle$) is obtained by induction on the premise.
 
($\co$) is obtained by induction on the lexicographically ordered pair (size of $A$, height of the proof of the premise), by looking at each possible last rule of the premise. The key case is ($\interl$):
\begin{prooftree}
  \AIC{\A,A,B,A\inter B,\B\vdash C}
  \ISCINTERL{\A,A\inter B,A\inter B,\B\vdash C}
\end{prooftree}
we apply ($\interle$) and ($\exc$) to the premise to get $\A,A,A,B,B,\B\vdash C$ and we use the induction hypothesis twice.
 
($\cut$) is obtained by induction on the lexicographically ordered triple (size of $A$, height of the proof of the left premise, height of the proof of the right premise), by looking at possible last rules of the premises. Let us focus on the main cases:
    \begin{itemize}
    \item ($\interr$) rule on the right:
      \begin{equation*}
\scalebox{0.7}{
        \AIC{\pi_1}
        \noLine
        \UIC{\A\vdash A}
        \AIC{\pi_2}
        \noLine
        \UIC{\B,A,\C\vdash B}
        \AIC{\pi_3}
        \noLine
        \UIC{\B,A,\C\vdash C}
        \ISCINTERR{\B,A,\C\vdash B\inter C}
        \ISCCUT{\B,\A,\C\vdash B\inter C}
        \DP
$\qquad\rightsquigarrow\qquad$
        \AIC{\pi_1}
        \noLine
        \UIC{\A\vdash A}
        \AIC{\pi_2}
        \noLine
        \UIC{\B,A,\C\vdash B}
        \ISCCUT{\B,\A,\C\vdash B}
        \AIC{\pi_1}
        \noLine
        \UIC{\A\vdash A}
        \AIC{\pi_3}
        \noLine
        \UIC{\B,A,\C\vdash C}
        \ISCCUT{\B,\A,\C\vdash C}
        \ISCINTERR{\B,\A,\C\vdash B\inter C}
        \DP}
      \end{equation*}
we use the induction hypothesis twice with a decreasing height on the right.
    \item ($\interr$) rule on the left and $(\interl)$ rule on the right:
      \begin{equation*}
\scalebox{0.7}{
        \AIC{\pi_1}
        \noLine
        \UIC{\A\vdash A}
        \AIC{\pi_2}
        \noLine
        \UIC{\A\vdash B}
        \ISCINTERR{\A\vdash A\inter B}
        \AIC{\pi_3}
        \noLine
        \UIC{\B,A,B,\C\vdash C}
        \ISCINTERL{\B,A\inter B,\C\vdash C}
        \ISCCUT{\B,\A,\C\vdash C}
        \DP
$\qquad\rightsquigarrow\qquad$
        \AIC{\pi_1}
        \noLine
        \UIC{\A\vdash A}
        \AIC{\pi_2}
        \noLine
        \UIC{\A\vdash B}
        \AIC{\pi_3}
        \noLine
        \UIC{\B,A,B,\C\vdash C}
        \ISCCUT{\B,A,\A,\C\vdash C}
        \ISCCUT{\B,\A,\A,\C\vdash C}
        \doubleLine
        \ISCEXC{\dotsb}
        \doubleLine
        \ISCCO{\B,\A,\C\vdash C}
        \DP}
      \end{equation*}
we use the induction hypothesis twice with smaller cut formulas.
    \item ($\constrr$) rules on both sides (in which we only write the key parts):
      \begin{gather*}
\scalebox{0.65}{
\AIC{\begin{array}{c}\dotsb\quad A_1\vdash A_1^i\quad\dotsb  \\ \vdots\\ \dotsb\quad A_{\convar{}}\vdash A_{\convar{}}^i\quad\dotsb\end{array}}
\AIC{\begin{array}{c}\dotsc,B_1^i,\dotsc\vdash B_1\\ \vdots \\ \dotsc,B_{\covar{}}^i,\dotsc\vdash B_{\covar{}}\end{array}}
\RL{\constrr}
\BIC{\dotsc,\constr(A_1^i,\dotsc,A_{\convar{}}^i;B_1^i,\dotsc,B_{\covar{}}^i),\dotsc\vdash \constr(A_1,\dotsc,A_{\convar{}};B_1,\dotsc,B_{\covar{}})}
\AIC{\begin{array}{c}\dotsb\quad C_1\vdash A_1\quad\dotsb  \\ \vdots\\ \dotsb\quad C_{\convar{}}\vdash A_{\convar{}}\quad\dotsb\end{array}}
\AIC{\begin{array}{c}\dotsc,B_1,\dotsc\vdash D_1\\ \vdots \\ \dotsc,B_{\covar{}},\dotsc\vdash D_{\covar{}}\end{array}}
\RL{\constrr}
\BIC{\dotsc,\constr(A_1,\dotsc,A_{\convar{}};B_1,\dotsc,B_{\covar{}}),\dotsc\vdash \constr(C_1,\dotsc,C_{\convar{}};D_1,\dotsc,D_{\covar{}})}
\ISCCUT{\dotsc,\constr(A_1^i,\dotsc,A_{\convar{}}^i;B_1^i,\dotsc,B_{\covar{}}^i),\dotsc\vdash \constr(C_1,\dotsc,C_{\convar{}};D_1,\dotsc,D_{\covar{}})}
\DP}
\\[2ex]\rightsquigarrow\qquad
\scalebox{0.65}{
\AIC{\dotsb}
\AIC{C_p\vdash A_p}
\AIC{A_p\vdash A_p^i}
\ISCCUT{C_p\vdash A_p^i}
\AIC{\dotsb}
\AIC{\dotsc,B_p^i,\dotsc\vdash B_p}
\AIC{\dotsc,B_p,\dotsc\vdash D_p}
\ISCCUT{\dotsc,B_p^i,\dotsc\vdash D_p}
\AIC{\dotsb}
\RL{\constrr}
\QuIC{\dotsc,\constr(A_1^i,\dotsc,A_{\convar{}}^i;B_1^i,\dotsc,B_{\covar{}}^i),\dotsc\vdash \constr(C_1,\dotsc,C_{\convar{}};D_1,\dotsc,D_{\covar{}})}
\DP}
      \end{gather*}
we use the induction hypothesis many times (always with smaller cut formulas).
\qedhere
    \end{itemize}
\end{proof}

Note a $0$-ary constructor $\constr$ behaves like an atomic type (\ie{} either a type constant or a type variable) if $\width{\constr}=1$, and defines a top type if $\width{\constr}=0$:
  \begin{prooftree}
    \AIC{\width{\constr}\leq 0}
    \RL{\constrr}
    \UIC{\vdash\constr}
    \ISCWKg{A\vdash\constr}
  \end{prooftree}
In particular the types obtained with such $0$-ary constructors $\constr$ such that $\width{\constr}=0$ are all equivalent and we denote them $\omg$.
More generally, $\width{\constr}$ controls whether $\constr$ distributes over $\omg$ or not. In the case of a constructor with unary covariant arity, $\width{\constr}$ determines whether $\constr(\vec{A};\omg)=\omg$ or not.

\begin{prop}[Kernel Properties]\label{propmin}
  If we define $A\leq B$ as $A\vdash B$ in \IScons, the axioms of Table~\ref{tabsub} are satisfied.
\end{prop}

\begin{proof}
(\ref{axsubrefl}) and (\ref{axsubtrans}) correspond to ($\ax$) and ($\cut$) from Proposition~\ref{propisconsadmiss}. (\ref{axsubinterr}) is an instance of ($\interr$) and for (\ref{axsubinterll}) we have:
\begin{prooftree}
  \ISCAX{A\vdash A}
  \ISCWKg{A,B\vdash A}
  \ISCINTERL{A\inter B\vdash A}
\end{prooftree}
Finally, if we have a $0$-ary constructor $\omg$ with $\width{\omg}=0$, we have just seen it satisfies (\ref{axsubomgr}).
\end{proof}

\begin{lem}[Inversion]\label{leminv}
  If $\constr(A_1^1,\dotsc,A_{\convar{\constr}}^1;B_1^1,\dotsc,B_{\covar{\constr}}^1),\dotsc,\constr(A_1^k,\dotsc,A_{\convar{\constr}}^k;B_1^k,\dotsc,B_{\covar{\constr}}^k)\vdash \constr(A_1,\dotsc,A_{\convar{\constr}};B_1,\dotsc,B_{\covar{\constr}})$, there exists $\{i_1,\dotsc,i_p\}\subseteq\{1,\dotsc,k\}$ such that:
  \begin{equation*}
    A_1\vdash A_1^{i_1} \;\dotsb\; A_1\vdash A_1^{i_p} \quad\dotsb\quad A_{\convar{\constr}}\vdash A_{\convar{\constr}}^{i_1} \;\dotsb\; A_{\convar{\constr}}\vdash A_{\convar{\constr}}^{i_p}
    \quad\text{and}\quad
    B_1^{i_1},\dotsc,B_1^{i_p}\vdash B_1 \quad\dotsb\quad B_{\covar{\constr}}^{i_1},\dotsc,B_{\covar{\constr}}^{i_p}\vdash B_{\covar{\constr}}
  \end{equation*}
\end{lem}

\begin{proof}
By induction on the derivation of $\constr(A_1^1,\dotsc,A_{\convar{\constr}}^1;B_1^1,\dotsc,B_{\covar{\constr}}^1),\dotsc,\constr(A_1^k,\dotsc,A_{\convar{\constr}}^k;B_1^k,\dotsc,B_{\covar{\constr}}^k)\vdash \constr(A_1,\dotsc,A_{\convar{\constr}};B_1,\dotsc,B_{\covar{\constr}})$, with only ($\wk$) and ($\constrr$) as possible last rules.
\end{proof}

\subsection{The Arrow-Product Instance}
\label{secbcdprod}

We consider the following set of constructors:
\begin{itemize}
\item an at most countable set of $0$-ary constructors denoted $X$, $Y$, etc, such that $\width{X}=\width{Y}=\dotsb=1$;
\item a $0$-ary constructor $\omg$ with $\width{\omg}=0$;
\item a constructor $\imp$ with contravariant arity $1$ and covariant arity $1$ such that $\width{\imp}=0$;
\item a constructor $\times$ with contravariant arity $0$ and covariant arity $2$ such that $\width{\times}=1$.
\end{itemize}

By instantiating the $(\constrr)$ rule of Table~\ref{tabiscons} to this set of constructors, and using the (\wk) rule to simplify the $X$ and $\omg$ cases, we obtain the rules of Table~\ref{tabisconsip} where $k\geq 1$.

\begin{table}
\centering
\begin{gather*}
\ISCAXX{X\vdash X}
\DP
\qquad
\ISCOMG{\vdash\omg}
\DP
\qquad
\AIC{A\vdash A_1}
\AIC{\dotsb}
\AIC{A\vdash A_k}
\AIC{B_1,\dotsc,B_k\vdash B}
\ISCIMP{A_1\imp B_1,\dotsc,A_k\imp B_k\vdash A\imp B}
\DP
\qquad
\AIC{\vdash B}
\ISCIMPZ{\vdash A\imp B}
\DP
\\[2ex]
\AIC{A_1,\dotsc,A_k\vdash A}
\AIC{B_1,\dotsc,B_k\vdash B}
\ISCTIMES{A_1\times B_1,\dotsc,A_k\times B_k\vdash A\times B}
\DP
\end{gather*}
\caption{\IScons{} Deduction System with $\imp$ and $\times$.}\label{tabisconsip}
\end{table}

\begin{thm}[Equivalence with \BCD]\label{thmisconsbcd}
  $A\vdash B$ in \IScons{} with the (\constrr) rule instantiated as given in Table~\ref{tabisconsip} if and only if $A\leq B$ using the rules of Table~\ref{tabsubbcd} extended with the rules of Table~\ref{tabsubbcdprod}.
\end{thm}

\begin{proof}
  From left to right, we prove a slightly more general statement: $A_1,\dotsc,A_k\vdash B$ implies $\biginter_{1\leq i\leq k}A_i\leq B$.
From right to left, the key results are in Propositions~\ref{propisconsadmiss} and~\ref{propmin}. Main cases are:
\begin{gather*}
  \ISCAX{C\vdash C}
  \ISCAX{C\vdash C}
  \ISCAX{A\vdash A}
  \ISCWKg{A,B\vdash A}
  \ISCAX{B\vdash B}
  \ISCWKg{A,B\vdash B}
  \ISCINTERR{A,B\vdash A\inter B}
  \RL{\imp}
  \TIC{C\imp A,C\imp B\vdash C\imp(A\inter B)}
  \ISCINTERL{(C\imp A)\inter(C\imp B)\vdash C\imp(A\inter B)}
  \DP
\qquad\qquad
  \ISCOMG{\vdash\omg}
  \ISCIMPZ{\vdash\omg\imp\omg}
  \ISCWK{\omg\vdash\omg\imp\omg}
  \DP
\\[3ex]
  \ISCAX{A\vdash A}
  \ISCWKg{A,C\vdash A}
  \ISCAX{C\vdash C}
  \ISCWKg{A,C\vdash C}
  \ISCINTERR{A,C\vdash A\inter C}
  \ISCAX{B\vdash B}
  \ISCWKg{B,D\vdash B}
  \ISCAX{D\vdash D}
  \ISCWKg{B,D\vdash D}
  \ISCINTERR{B,D\vdash B\inter D}
  \ISCTIMES{A\times B,C\times D\vdash (A\inter C)\times(B\inter D)}
  \ISCINTERL{(A\times B)\inter(C\times D)\vdash (A\inter C)\times(B\inter D)}
  \DP
\end{gather*}
\end{proof}

\begin{lem}[Inversion for Arrow]\label{leminvimp}
If $A\leq B$ is obtained from Tables~\ref{tabsubbcd} and~\ref{tabsubbcdprod}, we have:
  \begin{equation*}
  \biginter_{i\in I}(A_i\imp B_i)\leq A\imp B\;\Rightarrow\; \exists J\subseteq I,\;\left(\biginter_{i\in J}B_i\leq B\;\wedge\;\forall i\in J,\;A\leq A_i\right)
  \end{equation*}
\end{lem}

This is the key property of subtyping allowing for subject $\beta$-reduction to hold in the \BCD{} typing system. While the traditional proof goes by induction on the derivation which requires a more general statement to deal with the transitivity rule, we rely here on the subformula property. The traditional approach seems more difficult to use in a context where we may have many type constructors.

\begin{proof}
By Theorem~\ref{thmisconsbcd}, we have $\biginter_{i\in I}(A_i\imp B_i)\vdash A\imp B$, thus if $I=\{1,\dotsb,k\}$, we get $A_1\imp B_1,\dotsc,A_k\imp B_k\vdash A\imp B$ by Proposition~\ref{propisconsadmiss}. By applying Lemma~\ref{leminv}, we obtain $A\vdash A_{i_1}$,\dots, $A\vdash A_{i_p}$, $B_{i_1},\dotsc,B_{i_p}\vdash B$ with $J=\{i_1,\dotsc,i_p\}\subseteq I$, so that $\biginter_{i\in J}B_i\vdash B$, and we conclude with Theorem~\ref{thmisconsbcd}.
\end{proof}

\begin{lem}[Inversion for Product]\label{leminvprod}
If $A\leq B$ is obtained from Tables~\ref{tabsubbcd} and~\ref{tabsubbcdprod}, we have:
  \begin{equation*}
  \biginter_{i\in I}(A_i\times B_i)\leq A\times B\;\Rightarrow\;\biginter_{i\in I}A_i\leq A\;\wedge\;\biginter_{i\in I}B_i\leq B
  \end{equation*}
\end{lem}

\begin{proof}
Similarly by Theorem~\ref{thmisconsbcd}, Proposition~\ref{propisconsadmiss} and Lemma~\ref{leminv}.
\end{proof}

\subsection{\BCD{} Subtyping with Unary Constructors}

Our system \IScons{} also generalises \BCD{} subtyping with unary covariant constructors~\cite{intermixin}. In their setting constructors come as a set of unary covariant operations $\constr$ on types added to the usual $\imp$ and $\omg$ constructors:
  \begin{equation*}
        A,B ::= X \mid A\imp B \mid A\inter B \mid \omg \mid \constr(A)
  \end{equation*}
where each constructor $\constr$ satisfies the following subtyping properties:
\begin{equation*}
    \AIC{A\leq B}
    \UIC{\constr(A)\leq\constr(B)}
    \DP
\hspace{2cm}
    \ZIC{\constr(A)\inter\constr(B)\leq\constr(A\inter B)}
    \DP
\end{equation*}

This exactly corresponds, in the \IScons{} setting, to a set of constructors $\constr$ all satisfying $\convar{\constr}=0$, $\covar{\constr}=1$ and $\width{\constr}=1$ (the constructors $\imp$ and $\omg$ are obtained as before).
For example the associated ($\constrr$) rule can be derived in the \cite{intermixin} setting:
\begin{prooftree}
    \AIC{}
    \doubleLine
    \UIC{\biginter_{1\leq i\leq k}\constr(A_i)\leq \constr(\biginter_{1\leq i\leq k}A_i)}
    \AIC{\biginter_{1\leq i\leq k}A_i\leq A}
    \UIC{\constr(\biginter_{1\leq i\leq k}A_i)\leq\constr(A)}
    \BIC{\biginter_{1\leq i\leq k}\constr(A_i)\leq \constr(A)}
\end{prooftree}

\section{Conclusion}

We have presented a general way of defining a subtyping relation on intersection types which allows us to extend the \BCD{} subtyping to generic contravariant/covariant type constructors. It makes easy to derive key properties used to get subject reduction and subject expansion of the induced type systems. As a concrete example we have fully developed the extension of \BCD{} with product types.

Our approach can be extended to the case where a preorder relation $\cprec$ between constructors (with the same arities) leads to $\constr_1\cprec\constr_2\;\Rightarrow\;\constr_1(\vec A;\vec B)\leq\constr_2(\vec A;\vec B)$, by a natural generalisation of the ($\constrr$) rule:
\begin{prooftree}
\AIC{\begin{array}[b]{c}\constr_1\cprec \constr\\ \vdots\\ \constr_k\cprec \constr\end{array}}
\AIC{\begin{array}[b]{c}A_1\vdash A_1^1\quad\dotsb\quad A_1\vdash A_1^k \\ \vdots\\ A_{\convar{\constr}}\vdash A_{\convar{\constr}}^1\quad\dotsb\quad A_{\convar{\constr}}\vdash A_{\convar{\constr}}^k\end{array}}
\AIC{\begin{array}[b]{c}B_1^1,\dotsc,B_1^k\vdash B_1\\ \vdots \\ B_{\covar{\constr}}^1,\dotsc,B_{\covar{\constr}}^k\vdash B_{\covar{\constr}}\end{array}}
\AIC{\width{\constr}\leq k}
\QIC{\constr_1(A_1^1,\dotsc,A_{\convar{\constr}}^1;B_1^1,\dotsc,B_{\covar{\constr}}^1),\dotsc,\constr_k(A_1^k,\dotsc,A_{\convar{\constr}}^k;B_1^k,\dotsc,B_{\covar{\constr}}^k)\vdash \constr(A_1,\dotsc,A_{\convar{\constr}};B_1,\dotsc,B_{\covar{\constr}})}
\end{prooftree}
Ordering constructors is natural in the context of object-oriented languages~\cite{nominalsubtyping,intermixin}.

Another interesting instance would be the study of sum types, but subject expansion looks more complicated. From $\A\vdash\subst{u}{t}{x}:C$, we can find some type $A$ such that $\A\vdash t:A$ and $\A,x:A\vdash u:C$, but we do not find a way to complete the following derivation:
    \begin{prooftree}
    \AIC{\A\vdash t:A}
    \UIC{\A\vdash \texttt{inl }t:A+\_}
    \AIC{\A,x:A\vdash u:C}
    \AIC{\text{\Large ?}}
    \noLine
    \UIC{\A,y:\_\vdash v:C}
    \TIC{\A\vdash\begin{array}[c]{l}\texttt{match inl }t\texttt{ with}\\
                    \mid\; \texttt{inl } x \;\mapsto\; u \\
                    \mid\; \texttt{inr } y \;\mapsto\; v
                 \end{array} : C}
    \end{prooftree}
while \quad$\texttt{match inl }t\texttt{ with}\mid\texttt{inl } x \;\mapsto\; u \mid \texttt{inr } y \;\mapsto\; v$\quad reduces to \quad$\subst{u}{t}{x}$.
The idea of introducing some bottom type $\alpha$ with a rule \ZIC{\A,x:\alpha\vdash t:C}\DP seems too naive and breaks the system.

We also plan to work on the characterisation of normalizability properties of terms through typing properties in intersection type systems: solvability, normalization, strong normalization, etc. We would like to extend the known results~\cite{interbcd} to the case with more type constructors.

\paragraph{Acknowledgements.}
We would like to thank to Jan Bessai and Andrej Dudenhefner who suggested investigating \BCD{} with constructors.
Thanks also to the anonymous referees for their comments.

\bibliographystyle{eptcsalpha}
\bibliography{ll}
\end{document}